\documentclass[a4paper]{article}
\pdfoutput=1

\usepackage{amsmath, amsfonts, amssymb, bbm, physics}
\usepackage{enumitem}
\usepackage[hidelinks]{hyperref}

\usepackage{geometry}
\usepackage{amsthm, thmtools}
\usepackage[framemethod=tikz]{mdframed}
\usepackage[justification=centering]{caption}

\mdfsetup{skipabove=5pt,skipbelow=5pt}
\definecolor{thmcolor}{HTML}{3D3D3D}
\declaretheoremstyle[
    headfont=\bfseries, 
    bodyfont=\normalfont\itshape,
    headpunct={},
    mdframed={
        hidealllines=true,
        leftline=true,
        innertopmargin=0pt,
        linewidth=4pt,
        linecolor=thmcolor!50
    }
]{theoremstyle}

\declaretheorem[
    style= theoremstyle,
    name = Theorem,
    numberwithin = section
]{theorem}

\declaretheorem[
    style= theoremstyle,
    name = Lemma,
    sibling = theorem
]{lemma}

\definecolor{defcolor}{HTML}{AAAAAA}
\declaretheoremstyle[
    headfont=\bfseries, 
    headpunct={},
    mdframed={
        hidealllines=true,
        leftline=true,
        innertopmargin=0pt,
        linewidth=4pt,
        linecolor=defcolor!50
    }
]{definitionstyle}

\declaretheorem[
    style= definitionstyle,
    name = Definition,
    sibling = theorem
]{definition}

\renewenvironment{proof}
{ 
    \begin{mdframed}[
        hidealllines=true,
        leftline=true,
        innertopmargin=0pt,
        linewidth=1pt,
    ]
    \textit{Proof:}
}
{
    \qed
    \end{mdframed}
}

\newcounter{algorithm}
\newenvironment{algorithm}[2]
{
    \begin{mdframed}[
        linewidth=1pt,
    ]
    \refstepcounter{algorithm}
    \textbf{Algorithm~\thealgorithm:}\\
    
    \textbf{Input:} #1 \\
    \textbf{Output: } #2
    
    \begin{enumerate}
}
{
    \end{enumerate}
    \end{mdframed}
}

\usepackage[citestyle=numeric-comp, sorting=none, doi=true, isbn=false, url=false, eprint=false, maxbibnames=10]{biblatex}
\usepackage{authblk}

\author[1,*]{Wilfred Salmon}
\author[1]{Sergii Strelchuk}
\author[2]{Tom Gur}

\affil[1]{DAMTP, Centre for Mathematical Sciences, University of Cambridge}
\affil[2]{Department of Computer Science and Technology, University of Cambridge}
\affil[*]{was29@cam.ac.uk}

\usepackage{soul}

\definecolor{myblue}{HTML}{85B09A}

\usepackage{array}
\newcolumntype{M}[1]{>{\centering\arraybackslash}m{#1}}
\usepackage[]{xcolor}
\usepackage{colortbl}

\title{Provable Advantage in Quantum PAC Learning}

\newcommand{\lset}{\ensuremath{\mathcal{X}}}
\newcommand{\class}[1]{\ensuremath{\{0,1\}^{#1}}}
\newcommand{\ccc}{\ensuremath{\mathcal{C}}}

\newcommand{\dsb}{\ensuremath{\mathcal{D}}}
\newcommand{\dsbof}[1]{\ensuremath{\dsb(#1)}}
\newcommand{\prob}[2][]{\ensuremath{\mathbb{P}_{#1}\left[#2\right]}}
\newcommand{\dist}[2]{\ensuremath{d(#1, #2)}}

\newcommand{\orac}[1][c]{\ensuremath{Q_{#1}}}
\newcommand{\lrnst}[1][c]{\ensuremath{\ket{\psi_{#1}}}}
\newcommand{\inst}{\ensuremath{\ket{\text{IN}}}}

\newcommand{\refl}[1]{\ensuremath{I_{#1}}}
\newcommand{\ps}{\ensuremath{p_s(\theta)}}

\newcommand{\wmv}[2]{\ensuremath{\text{WMV}_{#1,\, #2}}}
\newcommand{\pho}{\ensuremath{\widetilde{O}_u}}

\begin{document}
\date{}
\maketitle

\begin{abstract}
We revisit the problem of characterising the complexity of Quantum PAC learning, as introduced by Bshouty and Jackson [SIAM J. Comput.
1998, 28, 1136–1153]. Several quantum advantages have been demonstrated in this setting, however, none are generic: they apply to particular concept classes and typically only work when the distribution that generates the data is known. In the general case, it was recently shown by Arunachalam and de Wolf [JMLR, 19 (2018) 1-36] that quantum PAC learners can only achieve constant factor advantages over classical PAC learners.

We show that with a natural extension of the definition of quantum PAC learning used by Arunachalam and de Wolf, we can achieve a generic advantage in quantum learning. To be precise, for any concept class $\ccc$ of VC dimension $d$, we show there is an $(\epsilon, \delta)$-quantum PAC learner with sample complexity
\[
O\left(\frac{1}{\sqrt{\epsilon}}\left[d+ \log(\frac{1}{\delta})\right]\log^9(1/\epsilon)\right).
\]
Up to polylogarithmic factors, this is a square root improvement over the classical learning sample complexity. We show the tightness of our result by proving an $\Omega(d/\sqrt{\epsilon})$ lower bound that matches our upper bound up to polylogarithmic factors.
\end{abstract}

\section{Introduction}
Probably approximately correct (PAC) learning \cite{Valiant1984} is a fundamental model of machine learning. One is given a set of functions $\ccc\subseteq \class{\lset}=\{f:\lset\to\{0,1\}\}$, called a concept class, that encodes the structure of a learning problem (for example, functions that only depend on the hamming weight of their input). Given labelled examples from an unknown concept $c\in\ccc$, we are tasked with learning an approximation to $c$. \\

We model the data that the learning algorithm receives by an unknown probability distribution $\dsb$ on $\lset$, and say that a hypothesis $h:\lset\to\{0,1\}$ is $\epsilon$-approximately correct if the probability that it differs from $c$ is at most $\epsilon$.
To be precise, A hypothesis $h\in\class{\lset}$ is said to be $\epsilon$-approximately correct if 
\begin{equation}
    \prob[X\sim\dsb]{h(X)\neq c(X)} \leq \epsilon.
\end{equation}
A learning algorithm $\mathcal{A}$ draws independent samples $(X, c(X))$, where $X$ is distributed according to $\dsb$, and then outputs a hypothesis $h$. The algorithm $\mathcal{A}$ is an $(\epsilon,\delta)$-learner if, with probability at least $1-\delta$ over the random samples, it outputs a $\epsilon$-approximately correct hypothesis.\\

The amount of ``structure" possessed by $\ccc$ is characterised by its Valiant-Chapernikis (VC) dimension \cite{Vapnik1971}, denoted $d$. For a subset $Y\subseteq X$, we define $\ccc|_Y:=\{c|_Y : c\in\mathcal{C}\}$ as the restriction of the concept class to $Y$. We say that $\ccc$ shatters $Y$ if $\ccc|_Y = \class{Y}$, i.e., if all possible labellings of $Y$ appear in concepts in $\ccc$. Then, $d$ is the maximum size of a shattered set, that is
\begin{equation}
    d = \max \{|Y| : Y \text{ is shattered by }\ccc\}.
\end{equation}

Over a period of 27 years \cite{Blumer1989, Hanneke2016}, the exact asymptotic scaling of the minimum number of samples required by an ($\epsilon, \delta$)-learner was found to be
\begin{equation}
    \Theta\left[\frac{1}{\epsilon}\left(d+ \log(\frac{1}{\delta})\right)\right],
\end{equation}
thereby characterising the complexity of classical PAC learning.\\

In 1995, Bshouty and Jackson \cite{Bshouty1995} considered a generalisation of PAC learning to the quantum setting \cite{Arunachalam2017}. Here, instead of receiving independent identically distributed samples $(X,C(X))$, one receives independent copies of a quantum state
\begin{equation}
    \lrnst = \sum_{x\in\lset} \sqrt{\dsbof{x}} \ket{x\; c(x)},
\end{equation}
known as a \emph{quantum sample}.
In particular, measuring such a state in the computational basis gives a sample $(X,C(X))$. In turn, instead of counting the number of samples, the quantum sample complexity is the number of copies of the state given to the quantum learning algorithm. \\

The Quantum PAC model is instrumental in understanding the limits of other quantum cryptographic and computational tasks. For instance, in \cite{arunachalam2021private}, a connection between differential privacy and PAC learnability of quantum states was established, and  recently \cite{cai2022sample} used the PAC framework to investigate the complexity of learning parameterised quantum circuits, which are ubiquitous in variational quantum algorithms where they are used for quantum state preparation. \\

In the special case of quantum PAC learning under the uniform distribution, it has been shown that one can obtain quantum sample complexity advantages in specific learning tasks, such as learning Fourier basis functions \cite{Bernstein1997}, DNF formulae \cite{Bshouty1995}, and $k$-juntas \cite{Chen2023}. These advantages rely on Fourier sampling, in which one applies the Hadamard transform on every qubit followed by a measurement of the resulting state in the computational basis. One observes a bit string $s$ with probability given by its squared Fourier coefficient $|\hat{c}_s|^2$ and can thus directly infer properties of the Fourier spectrum of the unknown function. However, such advantages rely on the distributions $\dsb$ being (approximately) uniform. \\

The general quantum PAC learning model, with an arbitrary and unknown distribution $\dsb$, was studied by Arunachalam and de Wolf \cite{Arunachalam2017, Arunachalam2018}, who showed that the quantum sample complexity has exactly the same asymptotic scaling as the classical learning complexity, ruling out everything but constant factor prospective advantages. \\

Thus, most recent literature has focused identifying advantages only in suitably restricted versions of the quantum PAC model \cite{Chen2023, Pirnay2023}. Nevertheless, such models have demonstrated remarkable utility when assessing the complexity of learning quantum states, channels \cite{aaronson2007learnability,chung2018sample,10.5555/3370256.3370257}, and measurements \cite{padakandla2022pac,cheng2015learnability} in quantum theory with lower bounds on query complexity established in \cite{zhang2010improved}.\\

Here, we consider a natural and less restrictive version of the quantum PAC learning model. Instead of access to copies of the state $\lrnst$, we assume that we have access to the quantum circuit that generates it, similarly in spirit to \cite{Kothari2023, vanApeldoorn2023}. That is, we assume one has access to a quantum circuit $\orac$ that generates a quantum sample $\lrnst$ (for example, as a decomposition into one and two-qubit gates) and thus can implement $\orac$ and $\orac^\dag$. Given this natural adjustment to the input access of quantum PAC learning algorithms, we can revisit the question of whether strong generic (beyond constant-factor) quantum advantages are possible for quantum PAC learning. \\

\subsection{Our results}
In this paper, we show that there is a square root advantage (up to polylogarithmic factors) for quantum PAC learning over classical PAC learning in the full, general model. Our main result (see Section \ref{sec:upper_bound}) is summarised by the following theorem.

\begin{theorem} 
    Let $\ccc$ be a concept class with VC dimension $d$. Then, for every $\epsilon,\delta>0$, there exists a $(\epsilon,\delta)$-quantum PAC learner for $\ccc$ that makes at most
    \begin{equation}\label{eqn:q_upp_bound}
        O\left(\frac{1}{\sqrt{\epsilon}}\left[d+ \log(\frac{1}{\delta})\right]\log^9(1/\epsilon)\right),
    \end{equation}
    calls to an oracle that generates a quantum sample ($\orac$) or its inverse ($\orac^\dag)$.
\end{theorem}

In comparison, the optimal classical PAC learning complexity (and quantum PAC complexity given access to copies of $\lrnst$ \cite{Arunachalam2018}) is given in equation \eqref{eqn:class_complexity}. Thus, our upper bound is a square root improvement (up to polylogarithmic factors) over the best possible classical learning algorithm. In fact, we show that this upper bound is essentially tight, up to polylogarithmic factors, as captured by the following theorem.

\begin{theorem}
    Let $\ccc$ be a concept class with VC dimension $d$. Then, for a sufficiently small constant $\delta>0$ and for all $\epsilon>0$, any quantum $(\epsilon,\delta)$-learner for $\ccc$ makes at least
    \begin{equation}\label{eqn:q_low_bound}
        \Omega\left(\frac{d}{\sqrt{\epsilon}}\right)
    \end{equation}
    calls to an oracle that generates a quantum sample ($\orac$) or its inverse ($\orac^\dag)$.
\end{theorem}

\subsection{Technical overview} 
Our starting point is the observation that the lower bound of Arunachalam and de Wolf \cite{Arunachalam2018} implicitly rests on the assumption that a quantum learning algorithm must not depend on the underlying concept, and it can thus be represented by a (concept independent) POVM. They then reduce the problem of PAC learning to that of state discrimination (where the POVM is state-independent). However, if we allow for the common assumption that the algorithm has access to an oracle $\orac$ generating $\lrnst$, the proof of the lower bound no longer holds\footnote{Since the state $\lrnst$ must be produced by some process, this assumption is quite minimal.}. If the POVM describing the algorithm calls the oracle, it, \textit{by definition}, depends on the underlying concept. Thus, one cannot reduce the problem to that of state discrimination, where it is assumed that the POVM is independent of the input state. \\ 

If one implements $\orac$ on some physical device (for example, as a series of one and two-qubit gates), it is natural to assume that one can also implement the inverse process $\orac^\dag$ (for example, by reversing the order of the gates and replacing each by its inverse). Thus, we argue that if one has access to the state $\lrnst$ it is natural to also consider the situation in which one also has access to $\orac$ and $\orac^\dag$. Indeed, this setting has recently received significant attention~\cite{vanApeldoorn2023, haah2023query}.\\

Given access to $\orac$ and $\orac^\dag$, it is tempting to attempt techniques such as Grover search and amplitude amplification, which often achieve quadratic quantum advantages. Consider, for example, the simplest possible concept class $\ccc=\class{\lset}$: the set of all possible classifiers. It is known that a classical worst-case distribution for this class is a ``perturbed" delta-function \cite{Arunachalam2018}, where there is a marked element $x_0\in \lset$ with probability $\dsbof{x_0} = 1-4\epsilon$, and all other elements have equal probability. Roughly speaking, to $(\epsilon,\delta)$-learn $\ccc$, one must learn a fraction of $3/4$ of the values of $c$. However, it takes on average $O(1/\epsilon)$ samples to return an $x$ that \textit{isn't} $x_0$ and thus the classical learning query complexity is $\Omega(|\lset|/\epsilon)$. In this case, one could repeatedly run Grover's search, marking any state $\ket{x\; b}$ as good if we have not yet learnt $c(x)$. With Grover search, it only takes $O(1/\sqrt{\epsilon})$ oracle calls to return an $x$ that is not $x_0$ and thus we see the quantum query complexity is $O(|\lset|/\sqrt{\epsilon})$, the desired quadratic improvement. Therefore, we already outperform the lower bound of Arunachalam and de Wolf \cite{Arunachalam2018}. \\

Note that the method above does not immediately generalise to other concept classes. For example, consider the concept class
\begin{equation*}
    \ccc=\{c\in\class{\lset}: |c^{-1}(\{1\})| = d\} \;,    
\end{equation*}
 the class of classifiers with exactly $d$ inputs that map to 1, and take $\dsb$ to be the uniform distribution on $\lset$. If $|\lset|$ is very large, then most unseen $x$'s will have $c(x)=0$ and thus the above approach is uninformative. Instead, instead, one should mark a state $\ket{x\; b}$ as good if $b = 1$. In this way, one can search for the inputs $x\in\lset$ that have $c(x)=1$ and hence deduce $c$. This will also lead to a quadratic quantum advantage. \\

However, for general concept classes, it is less clear what to search for. One could run the Halving algorithm, where we mark a state $\ket{x\; y}$ as good if the majority of the concepts $h\in \ccc$ that are consistent with the data so far have $h(x) = 1-y$. In this case, every time the Grover algorithm succeeds, one would eliminate at least half of the concepts in $\ccc$. However, this leads to a $\log |\ccc|$ factor in the learning complexity, which can be as large as $d\log |\lset|$, i.e., arbitrarily larger than $d$ (the VC dimension of $\ccc$). Thus, even under the simplifying assumption of the uniform distribution, it is unclear how to attempt to use Grover's search to obtain a quantum advantage. \\

Nevertheless, we show that one can achieve a square root quantum advantage in the general case. As a first step, we use the technique of equivalence queries \cite{Gluch2021} (also known as random counterexamples). An  equivalence query is an alternative learning oracle to the traditional PAC oracle, in which one submits a candidate hypothesis $h\in\class{\lset}$. If $h=c$, then the oracle outputs $YES$, otherwise it produces a labelled counterexample $(X, c(X))$ where
\begin{enumerate}[label=(\roman*)]
    \item {$h(X)\neq c(X)$.}
    \item {$X$ is distributed according to $\mathbb{P}(y) = \dsbof{y}/\dsbof{\{x:c(x)\neq h(x)\}}$.}
\end{enumerate}

Observe that by marking a state $\ket{x\; y}$ as good if $h(x) = 1-y$, we can see how to implement an equivalence query using Grover search, and thus one can hope to use this tool from classical learning theory to achieve an advantage. However, when one removes the simplifying assumption of a known distribution, further problems arise. \\

For a generic distribution, we do not know $\dsb(x)$ for any $x\in\lset$ and therefore one cannot run exact Grover search. Instead, we consider a well-studied technique \cite{Boyer1998}, in which one makes a random number of $M$ queries to the Grover oracle, where $M$ is uniformly distributed between 0 and a chosen threshold $T_G$. This search succeeds with non-negligible probability if the amplitude of the projection of the initial state onto the subspace spanned by the ``good" states (the ``good" subspace) is $\Omega(1/T_G)$. For an equivalence query $h$, this amplitude is $\sqrt{\dsbof{\{x:c(x)\neq h(x)\}}}$, which could be arbitrarily small (as $\mathcal{D}$ is arbitrary). Hence, it may take an arbitrarily large (expected) number of iterations of Grover's search (and hence oracle calls) to run a classical equivalence query learning algorithm. \\

To solve this issue, we show how to use equivalence queries that succeed with a constant probability, called imperfect equivalence queries, to PAC learn a concept. We can then run these imperfect equivalence queries using Grover search. We use a classical (ideal) equivalence query algorithm, replacing equivalence queries with repeated imperfect equivalence queries, but with a maximum imperfect equivalence query budget $R$. Suppose that the algorithm requires equivalence queries to hypotheses $h_1, \dots h_k$. If all of the successfully run an equivalence query for every hypothesis, then the classical algorithm succeeds, and we use its output. Otherwise, we hit the imperfect equivalence query budget $R$ and must terminate the classical algorithm early. By choosing $R$ sufficiently large, we can be sure that if we hit the budget, most of the imperfect equivalence queries were spent on hypotheses $h_i$ that are ``close" to $c$ (and hence have a low chance of the Grover search succeeding). Thus if we take the ``average" of the hypotheses $h_i$ weighted by the number of imperfect equivalence queries spent on each hypothesis, we also output a classifier close to $c$. \\

To conclude the section, we sketch a proof of our lower bound. We consider an arbitrary concept class $\ccc$ of VC dimension $d$. We note that there is a shattered set $Y\subseteq \lset$ of size $d$, and take $\dsb$ to be a ``perturbed" delta-function distribution on $Y$. We can thus think of concepts $c$ in $\ccc$ as bit strings of length $d$, where the bit string describes $c$'s action on $Y$. Since $Y$ is shattered by $\ccc$, all possible bit strings will appear. Any candidate PAC algorithm must be able to recover most of the bit string with high probability. We reduce to a known problem by introducing a weak phase-kickback oracle for the bit string, which we use to implement the PAC oracle. We can then use a standard lower bound \cite{vanApeldoorn2023} on recovering a bit string with high probability using a weak phase kickback oracle.

\subsection{Open problems}
This work leaves several interesting avenues for further research. Firstly, one could attempt to tighten the upper bound \eqref{eqn:q_upp_bound} to remove polylogarithmic factors and prove a tight matching lower bound. The removal of a $\log(1/\epsilon)$ factor in the query complexity for classical PAC learning took 27 years \cite{Blumer1989, Hanneke2016}; we hope that the quantum case will be simpler. Moreover, in order to achieve $1/\sqrt{\epsilon}$ scaling with our method, one would require the optimal classical equivalence query learning complexity to have no $\epsilon$ dependence and thus, a different approach is likely to be required. \\

It is interesting to consider the power of quantum learning algorithms with access to the oracle $\orac$, but not its inverse $\orac^\dag$. The inverse oracle seems necessary for Grover's search, and thus it is unclear if a quantum advantage is possible. The lack of such an advantage would have interesting implications for understanding what makes quantum computing more powerful than classical computation.\\

Finally, one could consider the implications of this work to generic advantages in more practical models of quantum machine learning, such as quantum neural networks.

\subsection{Organisation}
We first cover all required technical preliminaries in Section \ref{sec:prelims}. In Section \ref{sec:grover_sub}, we cover our Grover subroutine that leads to the quadratic advantage. Equivalence queries and how to use imperfect equivalence queries in a classical learning algorithm are both described in Section \ref{sec:equiv_queries}. Using the results of these two sections, we derive the upper bound \eqref{eqn:q_upp_bound} in Section \ref{sec:upper_bound}; we prove an almost matching lower bound on our quantum model in Section \ref{sec:lower_bound}, using a reduction to a phase oracle problem. Finally, we consider the application of our algorithm to learning $k-$juntas in Section \ref{sec:k-junta}.

\section{Preliminaries}\label{sec:prelims}
We will only consider functions defined on finite sets. We first introduce the standard, classical model of PAC learning \cite{Valiant1984}. For a finite set $\lset$, let $\class{\lset}=\{f:\lset\to\{0,1\}\}$, an element $f\in\class{\lset}$ is called a classifier. We wish to \textit{approximately} learn an unknown classifier $c$ from a known subset of classifiers $\ccc\subseteq\class{\lset}$, where $\ccc$ is called a concept class. \\

There is an unknown distribution $\dsb$ on $\lset$, where $\dsbof{x}$ denotes the probability of drawing $x$ from $\lset$. The distance between two classifiers is defined as the probability they disagree: $\dist{h_1}{h_2} := \prob[X\sim \dsb]{h_1(X)\neq h_2(X)}$. For a fixed tolerance $\epsilon>0$ we say a classifier $h\in\class{\lset}$ is $\epsilon$-\textit{approximately correct} if $\dist{h}{c}\leq \epsilon$. \\

A learning algorithm $\mathcal{A}$ has access to some oracle that gives information about $c$. Traditionally, one assumes that the oracle generates a labelled example $(X,c(X))$ at random, where $X$ is distributed according to $\dsb$. We will consider an additional type of oracle in section \ref{sec:equiv_queries}. The sample complexity of $\mathcal{A}$ is the number of labelled examples it receives.\\

For a fixed error probability $\delta$, we say that an algorithm $\mathcal{A}$ is an $(\epsilon, \delta)$ learner if, with probability at least $1-\delta$ (over the randomness of the algorithm), the algorithm outputs an $\epsilon$-approximately correct hypothesis, \textit{for every possible $c$ and $\dsb$}. \\

For a fixed concept class $\ccc$ and $\epsilon, \delta >0$, one wishes to find an $(\epsilon,\delta)$-learner with minimum sample complexity. The optimal sample complexity will depend on $\epsilon, \delta$ and some measure of complexity of the class $\ccc$, which we now define. For a subset $Y\subseteq \lset$, we define $\ccc|_Y:=\{c|_Y : c\in\mathcal{C}\}$ as the restriction of the concept class to $Y$. We say that $\ccc$ shatters $Y$ if $\ccc|_Y = \class{Y}$, i.e., if all possible labellings of $Y$ appear in concepts in $\ccc$. The Valiant-Chapernikis (VC) dimension \cite{Vapnik1971} of $\ccc$, denoted $d$, is the maximum size of a shattered set, that is
\begin{equation}
    d = \max \{|Y| : Y \text{ is shattered by }\ccc\}.
\end{equation}

In \cite{Blumer1989, Hanneke2016}, it was shown that the optimal sample complexity using labelled examples, denoted $T_C(\epsilon, \delta, d)$ scales as
\begin{equation}\label{eqn:class_complexity}
    T_C = \Theta\left[\frac{1}{\epsilon}\left(d+ \log(\frac{1}{\delta})\right)\right].
\end{equation}

In the quantum PAC setting \cite{Bshouty1995}, one assumes that the data is stored coherently, i.e., one considers the state 
\begin{equation}
    \lrnst:=\sum_{x\in \lset}\sqrt{\dsbof{x}}\ket{x\; c(x)},
\end{equation}
chosen so that measuring $\lrnst$ in the computational basis gives a random labelled example. Instead of the classical sample complexity, one considers the minimum number of copies $T_S(\epsilon, \delta, d)$ of $\lrnst$ required to PAC learn $\ccc$. Since one can always measure the state in place of a call to a classical oracle, $T_S$ is, at worst, the optimal sample complexity of a classical algorithm. In fact, Arunachalam and de Wolf \cite{Arunachalam2018} showed that there is no (asymptotic) quantum advantage from using states instead of oracle calls - the optimal $T_S$ grows exactly as in equation \eqref{eqn:class_complexity}. \\

We assume a stronger model, in which one has access to an oracle $\orac$ (which depends on the underlying concept), defined by its action on a fixed known input state $\inst$ (independent of the underlying concept):
\begin{equation}
    \orac \inst = \lrnst = \sum_{x\in \lset}\sqrt{\dsbof{x}}\ket{x\; c(x)}.
\end{equation}
This is similar in spirit to the recent work \cite{vanApeldoorn2023}, which deals with state tomography with a state preparation unitary. We also assume that the algorithm has access to the inverse of the oracle, $\orac^\dag$. This is relevant if, for example, $\orac$ is given as a quantum circuit of one or two-qubit gates; in this case, $\orac^\dag$ may be constructed by reversing the order of the gates and replacing each with its inverse. We define the learning complexity of any algorithm as the total number of queries to $\orac$ or $\orac^\dag$. The minimum learning complexity of any ($\epsilon, \delta$)-learner is denoted $T_O(\epsilon, \delta, \ccc)$. \\

The lower bound of \cite{Arunachalam2018} does not apply to a model with access to $\orac$, as it assumes the quantum algorithm is described by a POVM that is \textit{independent of the underlying concept $c$}. However, $\orac$ explicitly depends on $c$ and thus, any algorithm (or POVM) that calls $\orac$ will violate the assumptions in \cite{Arunachalam2018}. Hence, one can hope for quantum advantage in this setting. \\

We recap all of the different learning models considered in Table \ref{tab:learning_models}.\\

\begin{table}[ht]
    \centering
    \begin{tabular}{M{0.18\linewidth}|M{0.1\linewidth}|M{0.18\linewidth}|M{0.18\linewidth}|M{0.26\linewidth}}
    \rowcolor{lightgray}
         Model & Quantum or  Classical & Learning resource & Optimal $(\epsilon,\delta)$ learner complexity & Bounds on optimal learner complexity  \\ [8pt] \hline

         Labelled examples & Classical & Sample $(X,C(X))$ where $X\sim\dsb$ & $T_C$ & $\Theta\left[\frac{1}{\epsilon}\left(d+ \log(\frac{1}{\delta})\right)\right]$  \\ [12pt] \hline

         Equivalence queries & Classical & See Section \ref{sec:equiv_queries} & $T_E$ & $O\left(\left[d+\log(\frac{1}{\delta})\right]\log^9\left(\frac{1}{\epsilon}\right)\right)$ \\ [8pt] \hline

         Imperfect equivalence queries & Classical & See Section \ref{sec:equiv_queries} & $T_{IE}$ & $O(T_E)$ \\ [8pt] \hline

         Quantum samples & Quantum & Copy of $\lrnst$ & $T_S$ & $\Theta(T_C)$ \\ [8pt] \hline

         Quantum oracle calls & Quantum & Application of $\orac$ or $\orac^\dag$ & $T_O$ &  $O(\frac{1}{\sqrt{\epsilon}}T_{IE})$, $\Omega(\frac{d}{\sqrt{\epsilon}})$ 
    \end{tabular}
    \caption{Different learning models considered in our work. $T_M$ corresponds to the minimum number of resources needed by any $(\epsilon,\delta)$-learner in model $M$.}
    \label{tab:learning_models}
\end{table}

We end the preliminaries section with a recap of Grover's algorithm. For a subspace $\mathcal{V}$ of a Hilbert space $\mathcal{H}$, let $\Pi_{\mathcal{V}}$ be the orthogonal projection map onto $\mathcal{V}$. Furthermore, let $\refl{\mathcal{V}}$ be the reflection operator in $\mathcal{V}^\perp$, given by
\begin{equation}
    \refl{\mathcal{V}} = \mathbbm{1} - 2\Pi_{\mathcal{V}}.
\end{equation}
For a state $\ket{\psi}$, let $\refl{\ket{\psi}}$ be the reflection operator when $\mathcal{V}=$ span$\{\ket{\psi}\}$.\\

Grover search takes as its input a ``good" subspace $\mathcal{G}\subseteq\mathcal{H}$, and an input state $\ket{\psi}$. One then implements the Grover operator:
\begin{equation}\label{eqn:Grover_operator}
    D = - \refl{\ket{\psi}}\refl{\mathcal{G}}.
\end{equation}

The state $\ket{\psi}$ can be decomposed as 
\begin{equation}
    \ket{\psi} = \sin(\theta)\ket{g} + \cos(\theta)\ket{b},
\end{equation}
where $\ket{g},\ket{b}$ are orthonormal, $\theta\in[0,\pi/2]$, $\ket{g}\in\mathcal{G}, \ket{b}\in\mathcal{G}^\perp$. It is well-known \cite{Nielsen2010} that
\begin{equation}\label{eqn:angle_change}
    D^n \ket{\psi} = \sin((2n+1)\theta)\ket{g} + \cos((2n+1)\theta)\ket{b}.
\end{equation}
and thus if one knows $\theta$ exactly, one can apply $D^n$ such that $\sin((2n+1)\theta)\approx 1$.

\section{Grover Subroutine}\label{sec:grover_sub}
An essential subroutine for our quantum advantage is to use calls to $\orac$ and $\orac^\dag$ to run a Grover search \cite{Nielsen2010, Grover1996}. This leads to a quadratic improvement in learning complexity (up to polylogarithmic factors) over classical PAC learning. In this section, we describe our Grover subroutine.\\

Our Grover subroutine takes as an input a ``good" subset $G\subseteq\{(x,b):x\in \lset, b\in\{0,1\}\}$, where we wish to find an $x$ such that $(x,c(x))\in G$. We define a corresponding ``good" subspace by 
\begin{equation}
    \mathcal{G} = \text{span}\{\ket{x\; b} : (x,b)\in G\}.    
\end{equation}
In order to implement Grover's search, we need to implement the Grover operator, as defined in equation \eqref{eqn:Grover_operator}. We show that implementing $D$ requires a constant number of queries.

\begin{lemma}\label{lem:query_Grover}
    One can implement the Grover operator $D$ with one call to $\orac$ and one to $\orac^\dag$.
\end{lemma}
\begin{proof}
Note that $\refl{\mathcal{G}}$ is independent of $c$ and, therefore, may be implemented by a (possibly exponentially sized circuit) without any queries. To implement $\refl{\lrnst}$, note that

\begin{align}
    \refl{\lrnst} &= \mathbbm{1} - 2\dyad{\psi_c}, \\
              &= \orac(\mathbbm{1} - 2\dyad{\text{IN}})\orac^\dag, \\
              &= \orac\refl{\inst}\orac^\dag.
\end{align}
Note that $\refl{\inst}$ is independent of $c$ and, therefore, may be implemented by a (possibly exponentially sized circuit) without any queries. 
\end{proof}

We decompose 
\begin{equation}\label{eqn:angle_decomp}
    \lrnst = \sin(\theta)\ket{g} + \cos(\theta)\ket{b},
\end{equation}
where $\ket{g},\ket{b}$ are orthonormal, $\theta\in[0,\pi/2]$, $\ket{g}\in\mathcal{G}, \ket{b}\in\mathcal{G}^\perp$. If we knew $\theta$ exactly, we could apply $D^n$ such that $\sin((2n+1)\theta)\approx 1$. However, since $\theta$ depends on $\dsb$, which is unknown, this is impossible. Instead, we use the well-established \cite{Boyer1998} version of Grover's search for an unknown number of items. Our exact subroutine is given below; Algorithm \ref{alg:Grover_Search}.

\begin{algorithm}
    {$G \subseteq\{(x,b):x\in \lset, b\in\{0,1\}\}$ a good subspace, $\epsilon > 0$ a tolerance}
    {labelled example $(x,c(x))$. Succeeds if $(x,c(x))\in G$}
    \item{ Produce $\lrnst= \orac\inst$ }\label{alg:Grover_Search}
    \item{ Pick $N$ from $0,1\dots, \lceil2/\sqrt{\epsilon}\rceil-1$ uniformly at random }
    \item{ Apply $D$, the Grover operator, $N$ times to $\lrnst$ }
    \item{ Measure the resulting state in the computational basis }
\end{algorithm}

The properties of our algorithm are summarised in the following theorem
\begin{theorem}\label{thm:Grover_promises}
    Let $G \subseteq\{(x,b):x\in \lset, b\in\{0,1\}\}$ be a good subset, $\epsilon > 0$ be a fixed tolerance. Suppose that we run Algorithm \ref{alg:Grover_Search} with these inputs, then
    \begin{enumerate}[label=(\roman*)]
        \item{
            In the worst case, the algorithm makes $O(1/\sqrt{\epsilon})$ oracle (or inverse oracle) calls
        }
        \item {
            If $\prob[X\sim \dsb]{(X,c(X))\in G}\geq \epsilon$ then the algorithm succeeds, i.e., returns $(x,c(x))\in G$, with probability at least $p=0.09$.
        }
        \item {
            Conditional on succeeding, the output of the algorithm $(X,c(X))$ is distributed according to
            \begin{equation}
                \prob{(X,c(X))|\text{algorithm succeeds}} = \frac{\prob[X\sim\dsb]{X}}{\prob[X\sim\dsb]{(X,c(X))\in G}}.
            \end{equation}
        }
    \end{enumerate}
\end{theorem}
\begin{proof}\\
Part $(i)$: From the definition of the algorithm and Lemma \ref{lem:query_Grover}, the worst case number of oracle calls is $1+2(\lceil2/\sqrt{\epsilon}\rceil-1) = O(1/\sqrt{\epsilon})$.\\

Part  $(ii)$: Let $M = \lceil2/\sqrt{\epsilon}\rceil$, let $\theta$ be as in equation \eqref{eqn:angle_decomp} and let $\ps$ be the probability that the algorithm succeeds. Note that $\prob[X\sim \dsb]{(X,c(X))\in G}\geq \epsilon \Leftrightarrow \sin(\theta)\geq \sqrt{\epsilon}$. We use Lemma 2 (section 6) from \cite{Boyer1998}, which claims
\begin{equation}
    \ps = \frac{1}{2} - \frac{1}{4M}\frac{\sin(4M\theta)}{\sin(2\theta)}.
\end{equation}
For $\sin(\theta)\in[\sqrt{\epsilon}, 1/\sqrt{2}]$:\\
\begin{align}
    M &\geq \frac{2}{\sin(\theta)},\\
      &\geq \frac{1}{\sin(2\theta)},
\end{align}
and thus 
\begin{equation}
    \ps \geq \frac{1}{2}-\frac{1}{4} = \frac{1}{4} > 0.09.
\end{equation}
Note that for $\theta\in[\pi/4, \pi/2]$,
\begin{equation}
    \sin(2\theta)\geq\frac{\pi/2-\theta}{\pi/4},
\end{equation}
Thus for $\theta\in[\pi/4, (1/2-1/4M)\pi]$, we have that
\begin{align}
    p_s(\theta) &\geq \frac{1}{2} - \frac{1}{4M}\cdot \frac{4/\pi}{\pi/2 - (1/2-1/4M)\pi},\\
                &= \frac{1}{2} - \frac{4}{\pi^2} > 0.09.
\end{align}
Finally, for $\theta\in[(1/2-1/4M)\pi, \pi/2]$, note that $\sin(2\theta)\geq 0$ and $\sin(4M\theta) \leq 0$ so that $p_s(\theta) \geq 1/2 > 0.09$.\\

Part $(iii)$. This follows from the form of $D^n\lrnst$; the relative magnitude of the amplitudes in $\ket{g}$ is unchanged by the Grover operator $D$.
\end{proof}

We discuss how to combine the Grover subroutine with the algorithm of section \ref{sec:equiv_queries} to achieve a quantum learning complexity of equation \eqref{eqn:q_upp_bound} in section \ref{sec:upper_bound}. 

\section{Learning with imperfect equivalence queries}\label{sec:equiv_queries}
Equivalence queries are an alternative learning model for PAC learning. It was recently shown \cite{Gluch2021} that PAC learning with equivalence queries gives an exponential advantage over learning with labelled examples. In this section, we show how to use imperfect equivalence queries to PAC learn a concept class.

\begin{definition}
    An (ideal) equivalence query consists of submitting a candidate hypothesis $h$ for an underlying true concept $c$. If $h=c$ then we are told YES. Otherwise, we receive a labelled example $(x, c(x))$ where $c(x)\neq h(x)$ at random according to the distribution $\mathbb{P}(y) = \dsbof{y}/\dsbof{\{x:c(x)\neq h(x)\}}$. Such a labelled example where $h(x)\neq c(x)$ is called a counterexample.
\end{definition}

Equivalence queries are a very strong learning model, which is perhaps unrealistic. Thus, we assume we can only implement them probabilistically:\\

\begin{definition}
    An imperfect equivalence query consists of submitting a candidate hypothesis $h$ for the underlying concept $c$. In return we receive some labelled example $(x,c(x))$ with the following promises
    \begin{itemize}
        \item The distribution of $(X, c(X))$ \textit{conditional on being a counterexample} is the same as an ideal equivalence query.
        \item If $\dist{h}{c} \geq \epsilon$ then with some constant probability $p$ we receive a counterexample.
    \end{itemize}
\end{definition}

Note that we can tell whether our imperfect equivalence query failed or not - we can look at the result $(x, c(x))$ and check whether $h(x) = c(x)$. If they are equal, the equivalence query failed. Otherwise, it succeeded. Classically, we can implement an imperfect equivalence query using $1/\epsilon$ random labelled examples - we just sample $1/\epsilon$ times and see whether $c(x)\neq h(x)$ for any of our samples. On a quantum computer we can do this in $1/\sqrt{\epsilon}$ time using Grover's algorithm, as described in section \ref{sec:grover_sub} in Theorem \ref{thm:Grover_promises}.\\

We need one additional tool from classical learning theory to run our algorithm:\\

\begin{definition}
    Suppose we have a set of classifiers $\mathcal{H}\subseteq \class{\lset}$ and a distribution $\rho$ on $\mathcal{H}$. Then the weighted majority vote \cite{Masegosa2020}, $\wmv{\mathcal{H}}{\rho}\in\class{\lset}$ is defined such that it maximises
    \begin{equation}
        \prob[h\sim \rho]{\wmv{\mathcal{H}}{\rho}(x) = h(x)},
\end{equation}
    for every $x$ (ties can be broken arbitrarily).
\end{definition}

Suppose we have a classical algorithm $\mathcal{A}$ that uses $T_E(\epsilon, \delta, d)$ (ideal) equivalence queries to PAC learn a concept class $\ccc$. We show how to use $O(T_E+\log(1/\delta))$ imperfect equivalence queries to PAC learn the same concept class. \\

The full detail of the algorithm is given below in algorithm \ref{alg:imperfect_equiv_queries}. It works by running $\mathcal{A}$, replacing every equivalence query with repeated imperfect equivalence queries until one succeeds. We terminate if the learning algorithm $\mathcal{A}$ terminates or if we make a total of $R(T_E, \delta)$ imperfect equivalence queries. \\

We give some rough intuition for why the algorithm works before moving to prove so. If $\mathcal{A}$ terminates, then with high probability, it outputs an approximately correct hypothesis. If we pick $R$ large enough, then with high probability $T_E$ ideal queries to hypotheses $h_i$ with $\dist{h_i}{c}\geq \epsilon$ would all succeed in $<R/3$ imperfect equivalence queries. Thus, if the algorithm $\mathcal{A}$ does not terminate and we make $R$ total imperfect equivalence queries, with high probability, we spent $>2/3$ of our imperfect equivalence queries on hypotheses $h_i$ with $\dist{h_i}{c}<\epsilon$. Hence, if we take the weighted majority vote of all of the hypotheses we queried, weighted by the number of imperfect equivalence queries spent on each hypothesis, most of the vote will be decided by hypotheses that are close to the concept $c$. Thus, the weighted majority vote will also be close to $c$.\\

The full proof of why algorithm \ref{alg:imperfect_equiv_queries} works is given as two lemmas. Before these, we introduce some terminology.

\begin{definition}
    A transcript of a run of algorithm \ref{alg:imperfect_equiv_queries} is given by the list of hypotheses $\mathcal{H} = \{h_i\}$ that the algorithm queried along with a corresponding collection of natural numbers $n_i >0$, where $n_i$ is the number of imperfect equivalence queries spent on $h_i$.\\

    The time-spent distribution $\rho$ is the probability distribution on $\mathcal{H}$ given by $\rho(h_i) = n_i / \sum_i n_i$. \\

    Finally, $F = \{i: d(h_i, c) \geq \epsilon\}$ is called the ``feasible" set, where our imperfect equivalence query succeeds with probability at least $p$. Correspondingly $I = \{i : d(h_i, c) < \epsilon\}$\ is the ``infeasible" set, where there is no promise on the probability of success.
\end{definition}

Firstly, we show that with high probability that a bounded number of queries is spent on the feasible set\\

\begin{lemma}\label{lem:feasible_bound}
    With probability $\geq 1-\delta$ the total number of imperfect equivalence queries to feasible hypotheses is at most 
\begin{equation}
    2T_E/p + (1/2p^2)\log(1/\delta).
\end{equation} 
\end{lemma}
\begin{proof}
A imperfect equivalence query of a feasible hypothesis has (by definition) a chance $\geq p$ of succeeding, and the individual imperfect equivalence queries are independent. Additionally, there are at most $T_E$ feasible hypotheses to query (since the classical algorithm makes at most $T_E$ total equivalence queries). Thus, the probability that we succeed on all the feasible hypotheses using at most $m$ imperfect queries feasible hypotheses is lower bounded by the probability of getting at least $T_E$ successes from a binomial distribution $B(m, p)$. Thus, the chance of failure is lower bounded by the chance of fewer than $T_E$ successes from $B(m, p)$. \\

Let $X\sim B(m, p)$. Applying Hoeffding's inequality \cite{Hoeffding1963}, for $m\geq T_E/p$ we see that
\begin{equation}
    \prob{X<t} \leq e^{-2m(p-T_E/m)^2}.
\end{equation}
Thus it is sufficient for
\begin{equation}
    2m\left(p - \frac{T_E}{m}\right)^2 \geq \log(1/\delta).
\end{equation}
In turn, it is sufficient that
\begin{equation}
    2mp^2 - 4pT_E \geq \log(1/\delta),
\end{equation}
whence we deduce our bound. \\
\end{proof}

Next we prove that if we make enough imperfect equivalence queries on infeasible hypotheses, the weighted majority vote of the transcript must be close to the underlying concept $c$

\begin{lemma}\label{lem:maj_vote_good}
    Suppose we spend at least $2R/3$ imperfect equivalence queries on infeasible hypotheses. Then the weighted majority vote $M$ of the transcript with the time-spent distribution has $\dist{M}{c} < 4\epsilon$.
\end{lemma}
\begin{proof}
Fix the transcript $h_1, \dots h_k$. Let $\rho$ be the time-spent distribution and let $\rho'$ be the time-spent distribution conditioned on the infeasible set. That is, for $i\in I$, $\rho'(h_i) = \rho(h_i)/\rho(I)$. Similarly let $\tilde{\rho}$ the the time-spent distribution conditioned on the feasible set. We first show that if the infeasible set overwhelmingly votes for a bit $y$, then the whole transcript must also vote for that $y$. To be precise, suppose that $\prob[h\sim \rho']{h(x) = y} > 3/4$, then
\begin{align}
    \prob[h \sim \rho]{h(x) = y} &= \prob[h\sim \rho']{h(x)=y}\prob[h \sim \rho]{h\in I} + \prob[h\sim \tilde{\rho}]{h(x)=y}\prob[h \sim \rho]{h\in F},\\
    &> \frac{3}{4}\cdot\frac{2}{3}, \\
    &= \frac{1}{2}.
\end{align}
Letting $M=\wmv{\mathcal{H}}{\rho}$, we deduce (inspired by \cite{Masegosa2020}) that
\begin{align}
    \prob[X\sim \mathcal{D}]{M(X)\neq c(X)} &\leq \prob[X\sim \mathcal{D}]{\prob[h \sim \rho']{h(X)\neq c(X)}\geq \frac{1}{4}},\\
    \text{Markov's inequality,}& \leq 4 \mathbb{E}_{X\sim \mathcal{D}}\mathbb{E}_{h \sim \rho'}[\mathbbm{1}_{\{h(X)\neq c(X)\}}],\\
    & = 4\mathbb{E}_{h \sim \rho'}[\dist{h}{c}], \\
    \text{definition of infeasible set,}& < 4\epsilon
\end{align}
\end{proof}

We can now prove the performance of our algorithm

\begin{theorem}
    Let the maximum number of imperfect equivalence queries of algorithm \ref{alg:imperfect_equiv_queries} be 
    \begin{equation}
        R(T_E(\epsilon, \delta, d),\delta) = 6T_E(\epsilon, \delta, d)/p + + (3/2p^2)\log(1/\delta),
    \end{equation}
    then algorithm \ref{alg:imperfect_equiv_queries} produces a hypothesis $h$ with $\dist{h}{c} \leq 4 \epsilon$ with probability at least $1-2\delta$.
\end{theorem}
\begin{proof}
By Lemma \ref{lem:feasible_bound}, with probability $\geq 1- \delta$ we spend at most $R/3$ imperfect equivalence queries on feasible hypotheses - suppose this happens. If we succeed in an equivalence query for every hypothesis required by $\mathcal{A}$ then with probability at least $1-\delta$, $\mathcal{A}$ outputs a hypothesis $h$ with $\dist{h}{c} \leq \epsilon$. Otherwise, we spend at least $2R/3$ imperfect equivalence queries on infeasible hypotheses (as we assumed the feasible ones took at most $R/3$ imperfect equivalence queries) and then by Lemma \ref{lem:maj_vote_good} the weighted majority vote $\wmv{\mathcal{H}}{\rho}$ has $\dist{\wmv{\mathcal{H}}{\rho}}{c}< 4\epsilon$. Thus algorithm \ref{alg:imperfect_equiv_queries} outputs a $4\epsilon$-approximately correct hypothesis with probability at least $(1-\delta)^2 \geq 1- 2\delta$.
\end{proof}

\begin{algorithm}
    {$\delta >0, \epsilon > 0$ (the usual PAC parameters) and $\mathcal{A}$ a classical equivalence query learning algorithm with worst case query complexity $T_E>0$}
    {Hypothesis $h\in\class{\lset}$}\label{alg:imperfect_equiv_queries}
    {
        \item{Set the maximum imperfect equivalence query budget as $R = 6T_E/p  + (3/2p^2)\log(1/\delta)$. If $R$ total imperfect equivalence queries have ever been made, go to step 3}
        
        \item {Run $\mathcal{A}$, whenever it requires an equivalence query to a hypothesis $h$, repeatedly make imperfect equivalence queries until one succeeds. If $\mathcal{A}$ terminates, output the output of $\mathcal{A}$}
        
        \item{Let $\mathcal{H}=\{h_1,\dots,h_k\}$ be the set of hypothesis we ran imperfect equivalence queries on (so that $k\leq T_E$). Suppose we spent $n_i$ imperfect equivalence queries on $h_i$ (so that $\sum n_i = R$). Let $\rho(h_i) = n_i/N$ and output $h=\wmv{\mathcal{H}}{\rho}$}
    }
\end{algorithm}

\section{Upper bound on quantum learning complexity}\label{sec:upper_bound}
Here, we combine the results of sections \ref{sec:grover_sub} and \ref{sec:equiv_queries} to give an upper bound on $T_O$, the learning complexity of PAC learning with a state preparation oracle $\orac$ (and its inverse). \\

Suppose that it takes $E(\epsilon)$ queries to perform an imperfect equivalence query for a hypothesis $h$. If we have a classical equivalence learning algorithm $\mathcal{A}$ with a query complexity of $T_E(\epsilon, \delta, d)$, then we can use algorithm \ref{alg:imperfect_equiv_queries} of section \ref{sec:equiv_queries} to get a quantum PAC learning algorithm with learning complexity
\begin{equation}\label{eqn:implicit_ub}
    E(\epsilon/4)R(T_E(\epsilon/4, \delta/2, d),\delta/2).
\end{equation}
The current best known $T_E$ \cite{Gluch2021} has a worst-case query complexity of
\begin{equation}
    T_E = O\left(\left[d+\log(\frac{1}{\delta})\right]\log^9\left(\frac{1}{\epsilon}\right)\right).
\end{equation}
If we use the Grover subroutine (section \ref{sec:grover_sub} algorithm \ref{alg:Grover_Search}) with $G=\{(x,1-h(x)):x\in \lset\}$ to implement the imperfect equivalence queries, we find $E(\epsilon) = O(1/\sqrt{\epsilon})$. Substituting these $T_E$ and $E$ into the bound from equation \eqref{eqn:implicit_ub}, we get an upper bound of
\begin{equation}
    T_O = O\left(\frac{1}{\sqrt{\epsilon}}\left[d+ \log(\frac{1}{\delta})\right]\log^9\left(\frac{1}{\epsilon}\right)\right),
\end{equation}
which is a square-root improvement (up to polylogarithmic factors) over the classical PAC learning sample complexity of equation \eqref{eqn:class_complexity}.

\section{Lower bound on quantum learning complexity}\label{sec:lower_bound}
In this section, we prove a lower bound on quantum PAC learning with a state preparation oracle (and its inverse). We show that $\Omega(d/\sqrt{\epsilon})$ oracle calls are necessary. \\

Suppose we have a concept class $\ccc$ with VC dimension $d+1$. Then there is a set $Z$ of size $d+1$ in $\lset$ which is shattered by $\ccc$. We pick a marked element $x_0\in Z$ and let $Y=Z\setminus \{x_0\}$. We define our distribution $\dsb$ as a perturbed delta-function, the standard distribution used to prove lower bounds in learning:
\begin{equation}
    \dsbof{x} = \begin{cases}
        0, &\text{if } x\notin Z,\\
        1-4\epsilon, &\text{if } x = x_0,\\
        4\epsilon/d, &\text{if } x \in Y.
    \end{cases}
\end{equation}
We also restrict our concept class to $\widetilde{\ccc}=\{c\in\ccc:c(x_0) = 0\}$. If our PAC algorithm works on $\ccc$, it will certainly work on $\widetilde{\ccc}$. Since our distribution is restricted to $Z$ we need only identify the behaviour of our concept on $Z$. Thus, we can index our concepts by bit-strings $u\in\{0,1\}^d$ and index them with elements of $Y$. To be precise, we identify a concept $c\in\widetilde{\ccc}$ with a bit-string $u\in\{0,1\}^d$, where $u_y=c(y)$. \\

For a given bit-string $u\in\{0,1\}^d$, the state preparation oracle acts as
\begin{equation}
    \orac[u]\inst = \sqrt{1-4\epsilon}\ket{x_0\; 0} + \sqrt{\frac{4\epsilon}{d}}\sum_{x\in Y}\ket{x\; u_x}.
\end{equation}

Our main approach is to reduce to the following fact from Lemma 51 in \cite{vanApeldoorn2023}. \\
\begin{lemma}\label{lem:phase_oracle_bound}
    Let $u\in \{0,1\}^d$ be a bit string, and let $O_u$ be a weak phase-kickback oracle, that is
    \begin{equation}
        O_u \ket{x} = e^{2i\eta u_x}\ket{x}.
    \end{equation}
    Then recovering more than $3/4$ of the bits of $u$ with high probability requires at least $\Omega(d/\eta)$ calls to $O_u$, its inverse or controlled versions of these.\\    
\end{lemma}
\begin{proof}
    See \cite{vanApeldoorn2023}
\end{proof}

We will use calls to controlled versions of $O_u$ (denoted $c-O_u$) to implement the PAC state generation oracle $\orac[u]$. We fix $\eta\in[0,\pi/2]$ such that $\sin(\eta) = \sqrt{4\epsilon}$. \\

\begin{lemma}\label{lem:po_to_PACo}
    One can implement $\orac[u]$ using one call to $c-O_u$, one to $c-O_u^\dag$ and two qubit-ancillae. 
\end{lemma}
\begin{proof}
First, it is convenient to shift the phase to have a $\pm$ symmetry. Define a constant phase gate as\\
\begin{equation}
    P_\alpha \ket{x} = e^{i\alpha} \ket{x}.
\end{equation}
Then let
\begin{equation}
    \pho = P_{\eta}O_u^\dag,
\end{equation}
so that
\begin{equation}
    \pho \ket{x} = e^{i\eta\hat{u}_x}\ket{x},
\end{equation}
where
\begin{equation}
    \hat{u}_x = (-1)^{u_x}.
\end{equation}

We start by generating a uniform superposition of indices with the two-qubit ancillae in the $\ket{+}$ state:
\begin{equation}\label{eqn:intial_superposition}
    \frac{1}{2\sqrt{d}}\sum_{x\in Y}\ket{x}[\ket{00} + \ket{01} + \ket{10} + \ket{11}].
\end{equation}
We next apply 4 controlled gates - either $c-P_{\eta}$, $c-P_{-\eta}$ $c-\pho$ and $c-\pho^\dag$, such that each term in the superposition in equation \eqref{eqn:intial_superposition} picks up a different phase:
\begin{equation}
    \mapsto \frac{1}{2\sqrt{d}}\sum_{x\in Y}\ket{x}\left[e^{i\eta}\ket{00} + e^{-i\eta}\ket{01} + e^{i\eta\hat{u}_x}\ket{10} + e^{-i\eta\hat{u}_x}\ket{11}\right].
\end{equation}
Note that this requires two calls to singly controlled versions of the oracle - we can implement a double-controlled version by using a CCNOT (Toffoli) gate followed by a controlled oracle. Next, we apply a Hadamard gate to the second qubit register
\begin{equation}
    \mapsto \frac{1}{\sqrt{2d}}\sum_{x\in Y}\ket{x}\left[\ket{0}(\cos(\eta)\ket{0}+i\sin(\eta)\ket{1}) + \ket{1}(\cos(\eta \hat{u}_x)\ket{0}+i\sin(\eta\hat{u}_x)\ket{1}) \right].
\end{equation}
We then apply $S^\dag$ to the second qubit register (to remove the factors of $i$). We also use the even/odd ness of cos/sin to regroup the terms:
\begin{equation}
    \mapsto \frac{1}{\sqrt{2d}}\sum_{x\in Y}\ket{x}\left[\cos(\eta)(\ket{0}+\ket{1})\ket{0} + \sin(\eta)(\ket{0}+\hat{u}_x\ket{1})\ket{1}\right].
\end{equation}
We then apply a Hadamard gate to the first qubit register:
\begin{equation}
    \mapsto \cos(\eta)\left(\frac{1}{\sqrt{d}}\sum_{x\in Y} \ket{x} \right)\ket{0 0} + \sin(\eta)\left(\frac{1}{\sqrt{d}}\sum_{x\in Y} \ket{x\; u_x} \right)\ket{1}
\end{equation}
Conditional on the final qubit being in the state $\ket{0}$, we apply a unitary to the first register that maps the uniform superposition over $Y$ into the state $\ket{x_0}$:
\begin{equation}
    \mapsto \cos(\eta)\ket{x_0\; 0\; 0} + \sin(\eta)\left(\frac{1}{\sqrt{d}}\sum_{x\in Y} \ket{x\; u_x} \right)\ket{1}
\end{equation}
Finally, conditional on the first register not being in the state $\ket{x_0}$, we apply an $X$ gate to the second qubit register,followed by an $H$ gate on the second qubit register:
\begin{equation}
    \mapsto \left[\cos(\eta)\ket{x_0\; 0} + \sin(\eta)\left(\frac{1}{\sqrt{d}}\sum_{x\in Y} \ket{x\; u_x} \right)\right]\ket{+}
\end{equation}
But by the definition of $\eta$, we see that this is exactly equal to the action of the PAC oracle:
\begin{equation}
    (\orac[u]\inst)\ket{+}
\end{equation}
\end{proof}

We thus deduce our bound

\begin{theorem}
    $T_O = \Omega\left(\frac{d}{\sqrt{\epsilon}}\right)$
\end{theorem}
\begin{proof}
We can replace every call to $\orac[u]$ (or its inverse) in our PAC algorithm with the unitary process described in Lemma \ref{lem:po_to_PACo}, which requires a constant number of calls to (a controlled) $O_u$ (or its inverse). If the PAC algorithm outputs a correct hypothesis, then by construction of our distribution, it must agree on at least $3/4$ of the bits of $u$. Thus, the algorithm replaced with calls to $O_u$ (and its inverse) satisfies the conditions of Lemma \ref{lem:phase_oracle_bound}, and thus it must use at least $\Omega(d/\eta)$ calls to $O_u$. Hence, we reach a lower bound of 
\begin{equation}
    T_O = \Omega\left(\frac{d}{\arcsin{\sqrt{4\epsilon}}}\right) = \Omega\left(\frac{d}{\sqrt{\epsilon}}\right).
\end{equation}
\end{proof}

Note that our lower bound matches our upper bound (equation \eqref{eqn:q_upp_bound}), up to polylogarithmic factors.

\section{Application to learning $k-$juntas}\label{sec:k-junta}
A $k$-junta is a function $f:\{0,1\}^n\to {0,1}$ that only depends on a subset of $k$ bits. Letting $\lset=\{0,1\}^n$, we can consider the concept class $\ccc = \{f\in\class{\lset} : f \text{ is a }k\text{ junta}\}$. The exact VC dimension of $\ccc$ is unknown, but we can bound it using the inequalities
\begin{equation}
    2^d\leq |\ccc| \leq |\lset|^d + 1.
\end{equation}
The first of these comes from noting that if $\ccc$ shatters a set of size $\ell$, it must contain at least $2^\ell$ elements; the second is called Sauer's lemma \cite{Anstee2002}. We can bound
\begin{equation}
    |\ccc| \leq \binom{n}{k}2^{(2^k)},
\end{equation}
since there are $\binom{n}{k}$ ways to choose the $k$ bits determining the junta, and then $2^{(2^k)}$ choices for the underlying function. We deduce that
\begin{equation}
    d \leq \log\left[\binom{n}{k}\right] + 2^k \leq k\log(en/k) + 2^k.
\end{equation}
Thus, our learning algorithm can PAC learn a $k-$junta with
\begin{equation}
    O\left(\frac{1}{\sqrt{\epsilon}}\left[k\log\left(\frac{n}{k}\right) + 2^k + \log(\frac{1}{\delta})\right]\log^9(1/\epsilon)\right),
\end{equation}
oracle calls. This has a worse scaling in $n$ than algorihtms presented in \cite{Atici2007, Chen2023}, but has a better scaling in $\epsilon$ and works for \textit{any} underlying distribution, whereas previous work has focused on the uniform distribution.\\

\section*{Acknowledgements}
The authors thank J. van Apeldoorn, R. de Wolf, S. Arunachalam, J. Cudby, C. Long and J. Bayliss for helpful discussions related to this work.

Wilfred Salmon was supported by the EPRSC and Hitachi. Sergii Strelchuk acknowledges support from the Royal Society University Research Fellowship. Tom Gur is supported by the UKRI Future Leaders Fellowship MR/S031545/1 and an EPSRC New Horizons Grant EP/X018180/1. Sergii Strelchuk and Tom Gur are further supported by EPSRC Robust and Reliable Quantum Computing Grant EP/W032635/1.

\printbibliography
\end{document}